\documentclass[11pt,letterpaper]{article}

\usepackage{hyperref}
\hypersetup{
     colorlinks   = true,
     urlcolor    = blue,
	 citecolor = blue
}
\usepackage[margin=1in]{geometry}
\usepackage{amssymb,amsthm,amsmath,amssymb,wrapfig,dsfont,authblk}
\usepackage[dvipsnames]{xcolor}
\definecolor{myred}{RGB}{251,154,133}
\definecolor{myblue}{RGB}{153,206,227}
\definecolor{mylightblue}{RGB}{0, 150, 255}
\definecolor{mygreen}{RGB}{32, 210, 64}
\definecolor{mygray}{RGB}{220, 220, 220}

\usepackage{tikz}
\usetikzlibrary{decorations.pathmorphing}
\tikzset{snake it/.style={decorate, decoration=snake}}
\usetikzlibrary{shapes.geometric,positioning,decorations.pathreplacing}

\newcommand{\opt}{\small{\texttt{opt}}}

\newtheorem{theorem}{Theorem}[section]
\newtheorem{lemma}[theorem]{Lemma}
\newtheorem{claim}[theorem]{Claim}
\newtheorem{corollary}[theorem]{Corollary}

\newtheorem{definition}{Definition}[section]

\theoremstyle{definition}
\newtheorem{example}[theorem]{Example}

\newcommand{\E}{\mathbb{E}}
\newcommand{\R}{\mathbb{R}}
\newcommand{\M}{\mathcal{M}}
\newcommand{\X}{\mathcal{X}}
\newcommand{\ind}{\mathds{1}}

\allowdisplaybreaks

\begin{document}
\title{Opting Into Optimal Matchings}
\author[1]{Avrim Blum}
\author[2]{Ioannis Caragiannis}
\author[1]{Nika Haghtalab}
\author[1]{\\Ariel D. Procaccia}
\author[3]{Eviatar B. Procaccia}
\author[4]{Rohit Vaish}
\affil[1]{Computer Science Department, Carnegie Mellon University\\
	{\small\texttt{\{avrim,nhaghtal,arielpro\}@cs.cmu.edu}}}
\affil[2]{Department of Computer Engineering \& Informatics, University of Patras\\
	{\small\texttt{caragian@ceid.upatras.gr}}}
\affil[3]{Department of Mathematics, Texas A\&M University\\ 
	{\small\texttt{procaccia@math.tamu.edu}}}
\affil[4]{Department of Computer Science and Automation, Indian Institute of Science\\
	{\small\texttt{rohit.vaish@csa.iisc.ernet.in}}}

\date{}
\maketitle

\thispagestyle{empty}

\begin{abstract}
We revisit the problem of designing optimal, \emph{individually rational} matching mechanisms (in a general sense, allowing for cycles in directed graphs), where each player --- who is associated with a subset of vertices --- matches as many of his own vertices when he opts into the matching mechanism as when he opts out. We offer a new perspective on this problem by considering an arbitrary graph, but assuming that vertices are associated with players at random. Our main result asserts that, under certain conditions, \emph{any} fixed optimal matching is likely to be individually rational up to lower-order terms. We also show that a simple and practical mechanism is (fully) individually rational, and likely to be optimal up to lower-order terms. We discuss the implications of our results for market design in general, and kidney exchange in particular. 
\end{abstract}

\newpage
\setcounter{page}{1}

\section{Introduction}
\label{sec:intro}

Matching theory has made an astounding real-world impact, through the field of \emph{market design}; it is the cornerstone of the design and analysis of widely deployed applications that match residents to hospitals~\cite{RP99}, students to schools~\cite{APR05}, and organ donors to patients~\cite{RSU04,RSU05,RSU07}. But as matching markets become more prevalent, new issues arise, which potentially limit their (economic) efficiency. In this paper, we tackle one such issue: \emph{individual rationality} (or the lack thereof). Specifically, we study situations where the vertices of the graph are partitioned between a set of players, and each player is interested in matching as many of \emph{his own} vertices as possible. An \emph{individually rational matching} is one that matches at least as many vertices of each player as he can match on his own.

Why is individual rationality a real issue? Of the examples listed earlier, \emph{kidney exchange} provides arguably the most concrete, compelling answer. It is a medical innovation that, in its basic form, allows patients who need kidney transplant, and have willing but medically incompatible donors, to swap donors. From the matching viewpoint, the kidney exchange setting can be represented via a directed \emph{compatibility graph}, where each vertex corresponds to an incompatible \emph{patient-donor pair}, and there is an edge $(u,v)$ if the donor of $u$ is medically compatible with the patient of $v$. A pairwise swap corresponds to a 2-cycle in this graph, but exchanges along longer cycles --- and even along chains, initiated by altruistic donors --- are also important in practice (we also use the term \emph{matching} to refer to cycles and chains in these directed graphs). 

\begin{wrapfigure}{l}{0.4\textwidth}
\centering
\begin{tikzpicture}
\tikzstyle{redcirc}=[circle,
draw=black,fill=myred,thin,inner sep=0pt,minimum size=5mm]
\tikzstyle{bluecirc}=[circle,
draw=black,fill=myblue,thin,inner sep=0pt,minimum size=5mm]

\node (v1) at (0,0) [redcirc] {\small{$v_1$}};
\node (v2) at (1.5,0) [bluecirc] {\small{$v_2$}};
\node (v3) at (3,0) [bluecirc] {\small{$v_3$}};
\node (v4) at (4.5,0) [redcirc] {\small{$v_4$}};
\node (v5) at (0.75,1) [redcirc] {\small{$v_5$}};
\node (v6) at (2.25,1) [bluecirc] {\small{$v_6$}};
\node (v7) at (3.75,1) [redcirc] {\small{$v_7$}};

\draw [-latex] (v1) to (v2);
\draw [-latex] (v2) to (v5);
\draw [-latex] (v5) to (v1);
\draw [-latex] (v2) to (v3);
\draw [-latex] (v3) to (v6);
\draw [-latex] (v6) to (v2);
\draw [-latex] (v3) to (v4);
\draw [-latex] (v4) to (v7);
\draw [-latex] (v7) to (v3);
\end{tikzpicture}
\caption{A compatibility graph where individual rationality fails.}
\label{fig:notIR}
\end{wrapfigure}
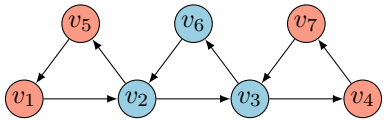 

Based on their work with practitioners, Ashlagi and Roth~\cite{AR14} have recently raised serious concerns regarding individual rationality in kidney exchange; they convincingly argue that as kidney exchange programs outgrow their regional origins, the incentives of hospitals (the players in this case) --- which have little to no interaction outside of the kidney exchange program --- become misaligned. In particular, hospitals cannot be certain that if they opt into a kidney exchange program, which optimizes overall efficiency, their patients would be better off overall than under the optimal \emph{internal} matching (which relies only on donor-patient pairs associated with the hospital). A bad example (due to Ashlagi and Roth) is given in Figure~\ref{fig:notIR}: the maximum cardinality matching selects the 3-cycles $\{v_1,v_2,v_5\}$ and $\{v_3,v_4,v_7\}$, but the blue player can do better by internally matching the single 3-cycle $\{v_2,v_3,v_6\}$. That is, the maximum cardinality matching is twice as large as the unique individually rational matching.

\subsection{Our Approach}

To summarize the preceding discussion, individual rationality is potentially a major obstacle to the economic efficiency of matching markets. Our goal is to analytically demonstrate that, in fact, individual rationality (or an almost perfect approximation thereof) can be achieved with nearly no loss of efficiency. Our key insight is that it suffices to assume that each vertex of the graph is owned by a \emph{random player}. 

In more detail, we consider an arbitrary graph with $n$ vertices $v_1,\ldots,v_n$, and a set of $k$ players $1,\ldots,k$. For each vertex, we draw its owner independently from the probability distribution $p_1,\ldots,p_k$ over the players, that is, each vertex is assigned to player $i$ with probability $p_i$. In the kidney exchange setting, for example, the rationale is very simple: the graph represents medical compatibility information, and there is no special reason why a patient-donor pair with particular medical characteristics would belong to a particular hospital --- the probability of that happening depends chiefly on the size of the hospital.

To see how randomization helps, let us revisit the example given in Figure~\ref{fig:notIR}, and suppose that the two players (red and blue) have probability $1/2$ each: $p_1=p_2=1/2$. The expected utility of a player under the maximum cardinality matching is $6/2=3$. In contrast, a straightforward upper bound on the expected cardinality of an internal matching can be derived by observing that each of the three 3-cycles is owned by a single player with probability $1/8$ and adds at most $3$ to the cardinality of the matching, leading to an upper bound of $3\cdot 3\cdot (1/8)=9/8$. Now, suppose we made $t$ copies of the graph of Figure~\ref{fig:notIR}, for a large $t$; then a simple measure concentration argument would imply that it is very likely that each player is better off in the optimal solution than he is working alone. 

Our goal is to establish this phenomenon in some generality. Indeed, our qualitative message (a few technical caveats apply) is that
\begin{quote}
\emph{... in an \emph{arbitrary} graph, under a random assignment of vertices to players, it is likely that \emph{any} fixed optimal matching is individually rational, up to lower order terms, for each player; and there is a \emph{practical} mechanism that yields an individually rational matching that is likely to be optimal up to lower order terms.}
\end{quote}

\subsection{Our Results and Techniques} 

In \S\ref{sec:opt}, we formalize the first part of the above statement. Specifically, we prove the following theorem:

\medskip

\noindent\textbf{Theorem \ref{thm:main} (informally stated).} \emph{Let $G$ be a directed graph, and let $\opt(G)$ be the set of vertices matched under a specific maximum cardinality matching on $G$. Assume that one of the following conditions holds:
\begin{enumerate}
\item Matchings are restricted to 2-cycles, and $p_i\leq 1/2$ for each player $i$, or
\item Matchings are restricted to cycles of constant length, and for each player $i$, $1/p_i$ is an integer. 
\end{enumerate}
Then for each player $i\in [k]$, the difference between the size of his optimal internal matching, and his share of $\opt(G)$, is at most $O(\sqrt{|\opt(G)|\cdot \ln(k/\delta)})$ with probability $1-\delta$.
}

\medskip

The theorem's first case deals with 2-cycles, a common abstraction for kidney exchange in theoretical studies~\cite{RSU05,TP15,ALG14,AFKP15,CFP15,BGPS13,BDHP+15,AKL16}. Of course in this case there always {\em exists} an optimal and individually rational matching (find the optimal internal matchings and then add augmenting paths), but nonetheless this statement is appealing because it applies to \emph{any} optimal solution that the exchange --- which might also be optimizing some secondary objective --- might produce. Also note that this case is essentially unrestrictive in terms of the probability distribution. The second case is the opposite: its assumption of constant length cycles is essentially unrestrictive, as chains can be represented as cycles by adding an edge from every patient-donor pair to every altruistic donor; and major kidney exchanges --- such as the US national program, run by the United Network for Organ Sharing (UNOS) --- use only cycles of length at most 3, and chains of length at most 4~\cite{RSU07,AR14,DPS12,AGRR12}. But the assumption regarding the probability distribution is, of course, somewhat restrictive. Note, however, that probabilities can be ``rounded'' at a cost, as we discuss later; and that the natural case of equal probabilities is captured by the second case. 

The proof of Theorem~\ref{thm:main} relies on two main ingredients. The first is the claim that the expected size of the maximum internal matching of player $i$ is at most a $p_i$ fraction of the optimal (global) matching. This statement is almost trivial in Case 2; to establish it in Case 1, we decompose the maximum cardinality matching via the Edmonds-Gallai Decomposition~\cite{Ed65}, and show that the inequality holds for each component separately.  

The second ingredient is the concentration of the cardinality of the optimal internal matching of each player around its expectation. To this end, we leverage machinery from modern probability theory that is little known in theoretical computer science, including a concentration inequality for so-called \emph{self-bounding functions}~\cite{boucheron2009concentration}.

The power of Theorem~\ref{thm:main} is that it applies to \emph{any} maximum cardinality matching. In the context of kidney exchange, the theorem captures the matching algorithms currently in use (including the ones employed by UNOS); its conceptual message is that hospitals need not worry about opting into kidney exchange programs, even under the \emph{status quo}. 

By contrast, in \S\ref{sec:ir} we give the designer more power in choosing the matching, with the goal of constructing a mechanism that is (perfectly) individually rational, and almost optimal. As noted earlier, this is quite trivial when only 2-cycles are allowed, as there always exists an optimal, individually rational matching. When longer cycles are allowed, we can derive the following corollary from the proof of Theorem~\ref{thm:main}.

\medskip

\noindent\textbf{Corollary~\ref{cor:ir} (informally stated).} \emph{Let $G$ be a directed graph with $n$ vertices, and let $\opt(G)$ be the set of vertices matched under a specific maximum cardinality matching on $G$. Suppose that matchings are restricted to cycles of constant length, and $p_i=p_j$ for any two players $i,j$. Then, with probability $1-\delta$, there exists a matching that is individually rational for each player $i$, and has maximum cardinality up to $O(k\sqrt{|\opt(G)|\cdot \ln(k/\delta)})$.}
\medskip

Importantly, such an individually rational and almost optimal matching can be found with a \emph{practical}\footnote{By ``practical'' we mean that it can be easily implemented in practice. It is not a polynomial-time algorithm, as computing a maximum cardinality matching is $\mathcal{NP}$-hard when 3-cycles are allowed~\cite{ABS07}; but the problem is routinely solved via integer programming.} mechanism: (i) compute a maximum cardinality matching, (ii) any player who wishes to work alone is allowed to defect. 

Furthermore, we show that our results are tight. Among other things, we construct an example with two players and cycles up to length 3 such that, with constant probability, \emph{any} individually rational matching is smaller than the optimal matching by $\Omega(\sqrt{|\opt(G)|})$.

\subsection{Related Work}

The two papers that are most closely related to ours are the ones by Ashlagi and Roth~\cite{AR14} and Toulis and Parkes~\cite{TP15}. Ashlagi and Roth show that under some technical assumptions, and under a random graph model of kidney exchange, large random graphs admit an individually rational matching that is optimal up to a certain constant fraction of the number of vertices, with high probability. Toulis and Parkes~\cite{TP15} independently study a very similar random graph model (it does make different assumptions about the size of hospitals), and obtain a similar result regarding individual rationality.

These important results have inspired our own work, but --- in addition to a number of significant technical advantages\footnote{In contrast to their work, we obtain optimality up to lower-order terms (instead of up to a constant fraction of $n$), and our results have a good dependence on the number of players (instead of assuming a very large~\cite{AR14} or a very small~\cite{TP15} number) and on the size of the graph (instead of assuming that $n$ goes to infinity).} --- we believe our high-level approach is significantly more compelling. In a nutshell, the random graph model studied by Ashlagi and Roth~\cite{AR14} and Toulis and Parkes~\cite{TP15} draws blood types for each donor and patient from a distribution that gives each of the four blood types (O, A, B, and AB) constant probability. For each pair of blood type compatible vertices (e.g., an O donor is blood type compatible with an A patient, but a B donor is not), a directed edge exists with \emph{constant} probability. This model clearly gives rise to very dense graphs; the key to the abovementioned results is that, with high probability, there exist matchings between blood type compatible groups (such as A patient and B donor, and B patient and A donor) that are perfect in the sense that they match all the vertices in the smaller group. Consequently, the structure of the optimal matching can be accurately predicted with high probability. This model has subsequently been employed in several other papers~\cite{DPS12,BGPS13}. 

However, more recent work by Ashlagi and Roth themselves --- together with collaborators~\cite{AGRR12} --- introduces a completely different random graph model of kidney exchange, which gives rise to sparse graphs, and better captures some real-world phenomena. This model was later employed by Dickerson et al.~\cite{DPS13}. At this point it is fair to say that, on the question of whether random graph models are a valid approach for the analysis of kidney exchange, the jury is still out. But we are convinced that an analysis that holds for \emph{arbitrary} graphs --- when it is feasible, as in this paper --- is the right approach.  

Individual rationality limits players to two possible strategies: work alone or participate fully. More generally, players can choose to reveal a subset of their vertices, and internally match the rest. Several papers seek to design mechanisms that incentivize players to reveal all their vertices, either as a dominant strategy~\cite{AFKP15,HDHS+15} or in equilibrium~\cite{AR14,TP15}. These known results are quite limited; obtaining stronger results is a central open problem. Our own approach does not seem to extend beyond individual rationality.

Needless to say, matching is a major research topic in theoretical computer science. In particular, there are many papers that are directly motivated by market design applications, especially kidney exchange~\cite{CTT12,CIKM+09,Adam11,GT12,BGLM+12,GN13}. These papers are largely orthogonal to our work. 

\section{Optimal Matchings Are Almost Individually Rational}
\label{sec:opt}
Designing and implementing new matching mechanisms can require significant changes to current policies and deployed algorithms. In this section, we show that even without any changes to the existing (optimal) matching mechanisms --- at least in the case of kidney exchange --- it is likely that each player matches almost as many vertices as what he could have obtained on his own.

Consider an arbitrary directed graph $G$ with $n$ vertices. Recall that for each player $i\in[k]$ with corresponding probability $p_i$, the player owns each vertex with probability $p_i$, independently. 
We denote by $H_i\sim_{p_i} G$ the random subgraph of player $i$, which is a subgraph of $G$ induced by assigning each vertex to $i$ with probability $p_i$. We suppress $p_i$ from this notation when it is clear from the context.
We use $\opt(G)$ to denote the set of \emph{vertices} of an arbitrary but fixed matching of $G$. Furthermore, $\opt(G) \restriction H_i$ denotes the restriction of $\opt(G)$ to subgraph $H_i$, i.e., the \emph{vertices} of $H_i$ that are matched under $\opt(G)$. Therefore, to compare the size of the internal matching of $H_i$ with the number of vertices of $H_i$ that are matched under the global matching, we compare  $|\opt(H_i)|$ to $|\opt(G)\restriction H_i|$, and show that these values are within $\tilde O(\sqrt{|\opt(G)|})$ of one another.

Let us first describe a graph in which, with a constant probability, a player's internal matching is larger by $\Omega(\sqrt{|\opt(G)|})$ than the player's share of any fixed optimal matching. 

\begin{example}
\label{ex:lb}
Suppose one of the players has probability $p= \frac 12$, and consider a graph that consists of $n/\log(n)$ stars, each with $\log(n)$ vertices that are connected to the center via $2$-cycles.
Fix an optimal global matching, $\opt(G)$, and note that $|\opt(G)| = 2n/ \log(n)$. We informally argue that there is a constant $c>0$ such that 
\[   \Pr_{H\sim_p G}\left[ | \opt(H) | - |\opt(G)\restriction H | > \Omega(\sqrt{|\opt(G)|})  \right]  > c.
\]

Indeed, let us consider the subgraph internal to the player, $H\sim_p G$.
While the expected number of centers in $H$ is $\frac 14 |\opt(G)|$, it is easy to see (by looking up the standard deviation of the binomial distribution) that, with constant probability, $H$ includes $t=\frac 14 |\opt(G)| + \Theta(\sqrt{|\opt(G)|})$ centers.
Moreover, with probability $\frac 12$, $H$ includes no more than half of the non-center vertices matched by $\opt(G)$. If both events occur, 
$| \opt(G)\restriction H | \leq t+\frac 14|\opt(G)|$,
where $t$ of the matched vertices correspond to the center vertices and at most $\frac 14 |\opt(G)|$ vertices correspond to non-center vertices of $H$ that coincide with $\opt(G)$.
On the other hand, each star is large enough so that with constant probability $H$ includes at least one non-center vertex in each star. In that case, for every center vertex in $H$, $\opt(H)$ gets two matched vertices. 
Therefore, $|\opt(H)| \geq  2t$ internally. It follows that the player can gain an additional $\Theta(\sqrt{|\opt(G)|})$ matched vertices when deviating from a fixed global optimal matching. 
\end{example}

Our main result shows that Example~\ref{ex:lb} is asymptotically tight.

\begin{theorem}
\label{thm:main}
Let $G$ be a directed graph and let $\opt(G)$ be the set of vertices matched under some fixed maximum cardinality matching on $G$. Assume that one of the following conditions holds:
\begin{enumerate}
\item Matchings are restricted to 2-cycles, and $p_i\leq 1/2$ for each player $i\in[k]$, or
\item Matchings are restricted to cycles of constant length, and for each player $i\in[k]$, $1/p_i$ is an integer. 
\end{enumerate}
Then for any $\delta>0$, 
\[\Pr_{H_i \sim_{p_i} G}\left[\forall i\in [k],\ | \opt(H_i) | - \left|\opt(G)\restriction H_i \right| < O\left( \sqrt{|\opt(G)|\ln \frac k\delta} \right)\right]\geq 1-\delta.
\] 
\end{theorem}

The proof of Theorem~\ref{thm:main} involves two main lemmas. The first shows that in expectation $|\opt(H_i)|$ is at most $|\opt(G)\restriction H_i |$. The second asserts
that $|\opt(H_i)|$ is concentrated nicely around its expectation. We formally state these two lemmas without further ado, but defer their proofs to \S\ref{sec:lem:pOPT} and \S\ref{sec:lem:concentrate}, respectively (with overflow in Appendix~\ref{app:main}). 

\begin{lemma}\label{lem:f(p)<pOPT}
Let $G$ be a directed graph and let $\opt(G)$ be the set of vertices matched under some fixed maximum cardinality matching on $G$. Then $  \E_{H\sim_p G}[ | \opt(H) |]  \leq p |\opt(G)|$ if (i) matchings are restricted to 2-cycles, and $p\leq 1/2$, or (ii) $1/p$ is an integer. 
\end{lemma}

\begin{lemma}\label{lem:concentration}
Let $G$ be a directed graph and let $\opt(G)$ be the set of vertices matched under some fixed maximum cardinality matching on $G$. Assume matchings are restricted to cycles of length up to a constant $L$.
Then for any $\delta>0$,  with probability $1-\delta$ over random choices of $H_i \sim_{p_i} G$, for all $i\in[k]$, 
\[  \E_{H_i\sim_{p_i} G}[ | \opt(H_i) |] -  L \sqrt{ 2\cdot \E[|\opt(H_i)|] \ln \frac{2k}{\delta} } < | \opt(H_i) |  <       \E_{H_i\sim_{p_i} G}[ | \opt(H_i) |] +  2L \sqrt{|\opt(G) |\ln \frac{2k}{\delta}}.
\]
\end{lemma}

We now easily prove our main result --- Theorem~\ref{thm:main} ---
by directly leveraging the two lemmas we just stated. 

\begin{proof}[Proof of Theorem~\ref{thm:main}]
Since $\opt(G)$ is fixed and $H_i$ is drawn from $G$ independently of $\opt(G)$, the expected number of vertices player $i$ has in $\opt(G)$ is
\begin{equation}\label{eq:EpOPT}
  \E_{H_i\sim_{p_i} G}[ | \opt(G) \restriction H_i| ] = p_i\cdot |\opt(G)|.
\end{equation}
Moreover, $| \opt(G) \restriction H_i| = \sum_{v\in \opt(G)} \ind_{v \in H_i}$, where $\ind_{v \in H_i}$ is an indicator variable with value $1$ if $v$ is owned by player $i$ and $0$ otherwise. So, $\ind_{v \in H_i}$ is a random variable that has value $1$ with probability $p_i$, and value $0$ otherwise. Using Hoeffding's inequality over $|\opt(G)|$ i.i.d.~variables for a fixed $i\in [k]$, as well as Equation~\eqref{eq:EpOPT},
\begin{equation}\label{eq:hoeffding}
  \Pr_{H_i\sim_{p_i} G}\left[| \opt(G) \restriction H_i| \geq p_i\cdot |\opt(G)| - \sqrt{\frac 12|\opt(G) | \ln \frac{2k}{\delta} }\right]\geq 1-\frac{\delta}{2k}.
\end{equation}
Putting this together with Lemmas~\ref{lem:f(p)<pOPT} and \ref{lem:concentration}, we have that for all $i\in[k]$, with probability $1-\delta$,
\begin{align*}
| \opt(H_i) |  & \leq  \E_{H_i\sim_{p_i} G}[ | \opt(H_i) |] +  2L \sqrt{|\opt(G) |\ln \frac{4k}{\delta}}\\
   & \leq     p_i~|\opt(G)|  +   2L \sqrt{ |\opt(G) |\ln \frac{4k}{\delta}}\\
   & \leq   | \opt(G) \restriction H_i|  +  (2L+1) \sqrt{|\opt(G) |\ln \frac{4k}{\delta}},
\end{align*}
where the first inequality follows from Lemma~\ref{lem:concentration} (using $\delta'=\delta/2$), the second from Lemma~\ref{lem:f(p)<pOPT}, and the third from applying Equation~\eqref{eq:hoeffding} to each $i\in [k]$.
\end{proof}

Note the logarithmic dependence of Theorem~\ref{thm:main} on the number of players $k$. A subtle point is that if the number of players is large, some will have a small $p_i$, which means that the expectation of $|\opt(G)\restriction H_i|$ is small compared to $|\opt(G)|$, by Equation~\eqref{eq:EpOPT}. In that case, a gain of $\sqrt{|\opt(G)|}$ is significant. Nevertheless, the theorem's conceptual message --- that following the global matching is individually rational up to lower order terms --- holds for any $p_i=\omega(1/\sqrt{|\opt(G)|})$. 

In addition, recall that Theorem~\ref{thm:main} considers two cases, (i) $1/p_i$ is an integer and (ii) $p_i\leq 1/2$ and $\opt(G)$ is restricted to  $2$-cycles.
Importantly, these two assumptions are only needed for Lemma~\ref{lem:f(p)<pOPT}. We conjecture that indeed Lemma~\ref{lem:f(p)<pOPT} holds for any $p\leq \frac 12$, whenever $\opt(G)$ is restricted to cycles of constant length --- in which case Theorem~\ref{thm:main}, too, would hold under this weaker assumption. One might wonder why assuming $p\leq \frac 12$ is even necessary. But in Appendix~\ref{app:large_p_example} we construct examples that violate the conclusion of Theorem~\ref{thm:main} for certain values of $p>1/2$. 

Finally, we remark that if $1/p_i$ is close to an integer but not itself an integer, one can first round down $p_i$ to the largest $q_i< p_i$ such that $1/q_i$ is an integer, and then apply Theorem~\ref{thm:main}. This would give the same result, up to an additional \emph{constant} fraction of $|\opt(G)|$. As $p_i$ becomes smaller the rounding error also diminishes. 

\subsection{Proof of Lemma~\ref{lem:f(p)<pOPT}} \label{sec:lem:pOPT}

First, let us address Case 2 of the lemma. Consider $p$ such that $1/p$ is an integer; $\opt(G)$ may include cycles of any length. Imagine there are $\frac 1 p$ players, each with probability $p$. By symmetry between the players, the expected size of the optimal matching in all subgraphs is equal. Furthermore, the total number of vertices matched by players individually is at most $\opt(G)$. Therefore, 
\[ |\opt(G) | \geq \sum_{i=1}^{1/p} \E_{H_i\sim_p G} \left[ | \opt(H_i) | \right]  = \frac 1 p \E_{H\sim_p G} \left[ | \opt(H) | \right],
\] 
which proves the claim.

In the remainder of this section we focus on Case 1 of Lemma~\ref{lem:f(p)<pOPT}, where the matchings are restricted to $2$-cycles and $p\leq 1/2$. For ease of exposition, we treat $G$ as an undirected graph: each directed 2-cycle corresponds to an undirected edge, and we may remove directed edges that are not involved in 2-cycles (as they are useless).

Assume there is a partition $G = G_1 \uplus G_2 \uplus \dots \uplus G_\ell$ of (the undirected graph) $G$ into edge-disjoint (but not necessarily vertex-disjoint) subgraphs that
preserve the size of the optimal matching, i.e., 
\begin{equation}
\label{eq:sum}
|\opt(G) | = \sum_{i= 1}^\ell |\opt(G_i)|.
\end{equation}
Moreover, assume that each of these subgraphs has the property that 
\begin{equation}
\label{eq:pOPT}
\E_{H\sim_p G_i}[ |\opt(H)| ] \leq p~ |\opt(G_i)|. 
\end{equation}
Then, the next equation proves that this property also holds for $G$ at the global level. That is,
\[\E_{H\sim_p G}[ | \opt(H) |] \leq  \sum_{i=1}^\ell \E_{H\sim_p G}\left[ | \opt(H \cap G_i) | \right] = \sum_{i= 1}^\ell \E_{H\sim_p G_i}[ | \opt(H) |] \leq \sum_{i=1}^\ell p ~ | \opt(G_i) | = p ~ | \opt(G)|.
\]
For the first transition, $H\cap G_i$ is the graph with edges that are present in both $H$ and $G_i$; the intuition behind this inequality is that we are essentially allowed to match the same vertices multiple times on the right hand side. The third and fourth transitions follow from Equations~\eqref{eq:sum} and \eqref{eq:pOPT}. 

So, it remains to find a partition of $G$ into edge-disjoint subgraphs, $G_1 \uplus G_2 \uplus \dots \uplus G_\ell$, which satisfies \eqref{eq:sum} and \eqref{eq:pOPT}. We prove that the Edmonds-Gallai Decomposition~\cite{Ed65} can be used to construct a partition satisfying these properties. 

\begin{lemma}[Edmonds-Gallai Decomposition]
Let $G=(V,E)$ be an undirected graph, let $B$ be the set of vertices matched by every maximum cardinality matching in $G$, and let
$D = V \setminus B$. Furthermore, partition $B$ into subsets $A$ and $C = B \setminus A$, where $A$ is the set of  vertices with at least one
neighbor outside $B$. And let $D_1, \dots, D_r$, be the connected components of the induced subgraph $G[D]$. Then the following properties hold.
\begin{enumerate}

\item $\opt(G)$ matches each node in $A$ to a distinct connected component of $G[D]$.
\item Each $D_i$ is \emph{factor-critical}, i.e., deleting any one vertex of $D_i$ leads to a perfect matching in the remainder of $D_i$.
\end{enumerate}
\end{lemma}

We now describe how the Edmonds-Gallai Decomposition is used to construct the desired partition of $G$. 
For the $i$th connected component of $G[C]$ and $G[D]$, create a subgraph $G_i$ corresponding to its edges.
Furthermore, for each vertex $i\in A$, create a subgraph $G_i$ corresponding to the set of edges incident on $i$. If there is an edge between two vertices of $A$, $i$ and $i'$, then include the edge in only one of $G_i$ or $G_{i'}$.
Since the Edmonds-Gallai Decomposition has no edges between $C$ and $D$, $G = \biguplus_{i} G_i$ forms a partition of the edge set of $G$.
See Figure~\ref{fig:decomposition} for an example of this construction.

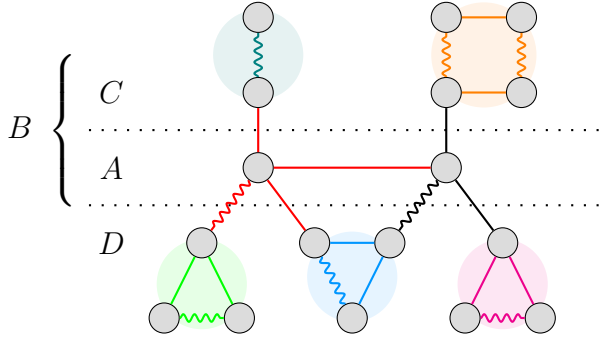
\begin{figure}[t]
\centering
\begin{tikzpicture}
\tikzstyle{graycirc}=[circle,
draw=black,fill=mygray,thin,inner sep=0pt,minimum size=4mm]
\tikzstyle{snake}=[-,thick,decorate,decoration={snake,amplitude=.5mm,segment length=1.5mm}]
\tikzstyle{line}=[-,thick]
\tikzstyle{group}=[circle,
draw=none,thin,inner sep=0pt, opacity=0.1]

\node (g1) at (0.25, 1.5) [group, minimum size =1.2cm, fill=teal] {};
\node (v1) at (0.25,1) [graycirc] {};
\node (v4) at (0.25,2) [graycirc] {};
\draw [snake, teal](v1) to (v4);

\node (g2) at (3.25, 1.5) [group, minimum size =1.4cm, fill=orange] {};
\node (v2) at (3.75,1) [graycirc] {};
\node (v5) at (3.75,2) [graycirc] {};
\node (v3) at (2.75,1) [graycirc] {};
\node (v6) at (2.75,2) [graycirc] {};
\draw [snake, orange](v2) to (v5);
\draw [snake, orange](v3) to (v6);
\draw [line, orange](v2) to (v3);
\draw [line, orange](v5) to (v6);

\node (u1) at (0.25,0) [graycirc] {};
\node (u2) at (2.75,0) [graycirc] {};

\node (g3) at (-0.5, -1.55) [group, minimum size =1.2cm, fill=green] {};
\node (g4) at (1.5, -1.45) [group, minimum size =1.2cm, fill=mylightblue] {};
\node (g4) at (3.5, -1.55) [group, minimum size =1.2cm, fill=magenta] {};

\node (w1) at (-0.5,-1) [graycirc] {};
\node (w2) at (1,-1) [graycirc] {};
\node (w3) at (2,-1) [graycirc] {};
\node (w4) at (3.5,-1) [graycirc] {};

\node (w5) at (-1,-2) [graycirc] {};
\node (w6) at (0,-2) [graycirc] {};
\node (w7) at (1.5,-2) [graycirc] {};
\node (w8) at (3,-2) [graycirc] {};
\node (w9) at (4,-2) [graycirc] {};

\draw [line, red](u1) to (u2);
\draw [line, red](u1) to (v1);
\draw [line, red](u1) to (w2);
\draw [snake, red](u1) to (w1);

\draw [line, black](u2) to (v3);
\draw [line, black](u2) to (w4);
\draw [snake, black](u2) to (w3);

\draw [line, green](w1) to (w5);
\draw [line, green](w1) to (w6);
\draw [snake, green](w5) to (w6);

\draw [line, mylightblue](w2) to (w3);
\draw [line, mylightblue](w3) to (w7);
\draw [snake, mylightblue](w7) to (w2);

\draw [line, magenta](w4) to (w8);
\draw [line, magenta](w4) to (w9);
\draw [snake, magenta](w8) to (w9);

\draw[loosely dotted, thick] (-2 ,0.5) -- (5, 0.5); 
\draw[loosely dotted, thick] (-2 ,-0.5) -- (5, -0.5); 

\node at (-1.7,0) {\large $A$};
\node at (-1.7,1) {\large $C$};
\node at (-1.7,-1) {\large $D$};
\node at (-2.4,0.5) {\large $B ~ \begin{cases} \vspace{1.6cm}\end{cases}$};

\end{tikzpicture}
\caption{\small A graph demonstrating the Edmonds-Gallai Decomposition and the edge-disjoint graph partition for the proof of Lemma~\ref{lem:f(p)<pOPT}. In this graph, each color represents one $G_i$ in the partition $G = \biguplus_i G_i$ and the wavy edges represent the matched edges in $\opt(G)$.}
\label{fig:decomposition}
\end{figure}

We argue that the foregoing partition satisfies Equation~\eqref{eq:sum}; the proof of this claim is relegated to Appendix~\ref{app:sum}.

\begin{claim}
\label{claim:sum}
$| \opt(G) | = \sum_{i} |\opt(G_i)|$.
\end{claim}

Next, we show that $G = \biguplus_{i} G_i$ satisfies \eqref{eq:pOPT}. There are three types of $G_i$ in this partition: (i)~$G_i$ is a star, (ii) $G_i$ is a component of $G[D]$ and has a matching that covers all but one vertex, and (iii) $G_i$ is a component of $G[C]$ and has a perfect matching. 

Let us first address case (i) --- that of a star. Clearly it holds that $|\opt(G_i)|=2$. Now, $|\opt(H)|\in \{0,2\}$, and for $\opt(H)$ to be non-empty, $H$ must include the center of the star, which happens with probability $p$.

The following claim, whose proof is relegated to Appendix~\ref{app:type23}, establishes Equation~\eqref{eq:pOPT} in cases (ii) and (iii). Note that this is the only place where the assumption $p\leq 1/2$ is used.

\begin{claim} \label{claim:type23}
For any $p\leq \frac 12$, and any graph $G$ with $n$ vertices such that $|\opt(G)|\geq n-1$,
\[\E_{H\sim_p G}\left[ |\opt(H)| \right] \leq p~|\opt(G)|.\]
\end{claim}

Having established \eqref{eq:pOPT}, the proof of Lemma~\ref{lem:f(p)<pOPT} is now complete. \qed

\subsection{Proof of Lemma~\ref{lem:concentration}} \label{sec:lem:concentrate}
We will prove the following equivalent formulation of Lemma~\ref{lem:concentration}:
\begin{equation}
\label{eq:upper}
\Pr_{H\sim G} \left[ |\opt(H)| \geq  \E_{H\sim G}[|\opt(H)|] + \epsilon \right] \leq \exp\left(-\frac{\epsilon^2}{4 L^2~ |\opt(G)|}\right),
\end{equation}
and
\begin{equation}
\label{eq:lower}
\Pr_{H\sim G} \left[ |\opt(H)| \leq  \E_{H\sim G}|[\opt(H)|] - \epsilon \right] \leq \exp\left(-\frac{\epsilon^2}{2 L^2~ \E[|\opt(H)|]}\right).
\end{equation}

Let us first describe a failed approach for proving the lemma, which brings to light some subtleties in the above inequalities.
Consider an explicit description of $|\opt(H)|$ as a function of $n$ random variables, $X_1, \dots, X_n$, where $X_i = 1$ if vertex $i$ is in $H$ and $0$ otherwise. Then $|\opt(H)| = f(X_1, \dots, X_n)$ is the size of the optimal matching on $H$.
One can show that $f(\cdot)$ is $L$-Lipschitz, that is, changing $X_i$ to $\neg X_i$, which corresponds to adding or removing one vertex from $H$, changes the size of the maximum matching by at most $L$. Lipschitz functions are known to enjoy strong concentration guarantees, as shown by McDiarmid's inequality, 
\[\Pr \left[ | f  -  \E[f] | >  \epsilon \right] \leq 2 \exp \left( \frac {-2\epsilon^2}{\sum_{i=1}^n c_i^2} \right), 
\]
where $c_i$ is the Lipschitz constant for the $i$th variable, that is, for all $i$ and for every possible input $x_1, \dots, x_n$, $| f(x_1, \dots, x_i, \dots, x_n) -  f(x_1, \dots, \neg x_i, \dots, x_n)| \leq c_i$. 

While there are only $|\opt(G)|$ variables that truly participate in $\opt(G)$, even vertices that are not in $\opt(G)$ can participate in 
matchings of subsets of $G$, and as a result have a non-zero Lipschitz constant.  Therefore, using McDiarmid's inequality for the concentration of $\opt(H)$ gives
an $O(\sqrt{n})$ gap between $\opt(H)$ and $\E_{H\sim G}[\opt(H)]$. 

Instead, in order to prove a gap of $\tilde{O}(\sqrt{|\opt(G)|})$, we use two alternative concentration bounds from statistical learning theory, which have recently been used to simplify and prove tight concentration and sample complexity results for learning combinatorial functions~\cite{vondrak2010note}.

\begin{lemma}\emph{\cite[Theorem 12]{boucheron2004concentration}} \label{prop:concentrate}
Let $X_1, X_2, \dots, X_n$ be independent random variables, each taking values in a set $\X$.
Let $f : \X^n \rightarrow \mathbb{R}$ be a measurable function.
Let $X'_1, \dots, X'_n$ be independent copies of $X_1, \dots, X_n$ and for all $i\in [n]$,
define $f'_i = f(X_1, \dots, X'_i, \dots, X_n)$. 
For all $i\in [n]$ and $\mathbf{x} \in \X^n$ assume that there exists $C>0$, such that 
\[  \E \left[\left. \sum_{i=1}^n \left(f - f'_i \right)^2 \cdot \ind_{f > f'_i} \right| \mathbf{x} \right] \leq C, 
\]
then for all $\epsilon >0$,
\[   \Pr[f > \E[f] + \epsilon] \leq e^{-\epsilon^2/4C}.
\]
\end{lemma}

We show that the conditions of Lemma~\ref{prop:concentrate} hold for $f(x_1, \dots, x_n) = \opt(H)$. 
Let $\mathbf{x}= (x_1, \dots, x_n)$ and $\mathbf{x}'_i = (x_1, \dots x'_i, \dots, x_n)$. For all $\mathbf{x}$,
let $H_{\mathbf x}$ be the subgraph corresponding to non-zero variables of $\mathbf{x}$.
Note that if $x_i$ is replaced by $x'_i$ and the matching size is reduced, then the decrease is at most the maximum cycle length $L$. 
Furthermore, the only variables that can lead to a non-zero decrease from $|\opt(H_{\mathbf{x}})|$ to $|\opt(H_{\mathbf{x}'_i})|$ are variables that 
are in every optimal matching on $H_{\mathbf{x}}$. Therefore, there are at most $|\opt(H_{\mathbf{x}})| \leq |\opt(G)|$ such variables. 
We conclude that for all $\mathbf{x}$,
\[ \E \left[\left. \sum_{i=1}^n \left(f - f'_i \right)^2 \cdot \ind_{f > f'_i} \right| \mathbf{x} \right] \leq 
   L^2 ~ |\opt(G)|. 
\]
The proof of the upper tail \eqref{eq:upper} follows immediately by using Lemma~\ref{prop:concentrate} with $C =  L^2 ~ |\opt(G)|$.

Unfortunately, Lemma~\ref{prop:concentrate} and its variants for lower-tail concentration cannot be used to establish the desired lower-tail bound \eqref{eq:lower}. Indeed, consider the condition $\E [ \sum_{i=1}^n \left(f - f'_i \right)^2 \cdot \ind_{f < f'_i} | \mathbf{x} ] \leq C$; while removing one of only $|\opt(G)|$ vertices can reduce the size of a matching, it may be possible that for some subgraph of $G$,  adding any of the remaining vertices increases the size of the matching. Instead, we use the lower-tail concentration of self-bounding functions~\cite{boucheron2009concentration}. The rigorous proof of Equation~\eqref{eq:lower} appears in Appendix~\ref{app:lower}. \qed

\section{Individually Rational Matchings That Are Almost Optimal}
\label{sec:ir}
In this section, we provide a \emph{simple and practical} mechanism for kidney exchange that guarantees individual rationality, and with high probability yields a matching that is optimal up to lower-order terms. In comparison to the results of \S\ref{sec:opt}, its disadvantage is that it requires modifying deployed matching mechanisms, which simply return some optimal matching --- our mechanism selects a \emph{specific} matching (which may be suboptimal). However, it is only a minor modification, and therefore has the potential to inform practice.

Let us first consider the case where $\opt(G)$ is restricted to $2$-cycles. In this case we can represent $G$ as an undirected graph, as in \S\ref{sec:lem:pOPT}. Let $H_1, \dots, H_k$ be the subgraphs corresponding to the players. Consider the following matching mechanism, $\M(H_1, \dots, H_k)$: First, compute the matching $M = \bigcup_i \opt(H_i)$; then grow $M$ to a globally maximum cardinality matching by repeatedly applying augmenting paths.
While an augmenting path changes the structure of a matching by adding and removing edges, it strictly expands the set of matched vertices. Therefore, this mechanism leads to a maximum cardinality matching on $G$, with the property that $\M(H_1, \dots, H_k) \supseteq \bigcup_i \opt(H_i)$ for all $i\in[k]$. That is, $\M$ is individually rational.

The performance of the above mechanism for $2$-cycles holds even when the subgraphs owned by players are chosen adversarially, rather than through a random process. Furthermore, this mechanism enjoys the stronger guarantee that every vertex that is matched under $\opt(H_i)$ is also matched under $\M(H_1, \dots, H_k)$.
As we discussed earlier (see Figure~\ref{fig:notIR}), these strong guarantees are unattainable when cycles of length $3$ are allowed. But in our model for randomly generating $H_1,\ldots,H_k$, there is a mechanism that is individually rational and almost optimal, as we show next.

\begin{corollary}
\label{cor:ir}
Let $G$ be a directed graph. Consider optimal matchings on $G$ that are restricted to constant-length cycles.  For all $i\in [k]$, let $p_i = 1/k$.
Then there exists a mechanism $\M$ such that $\M(H_1, \dots, H_k)$ is individually rational, and for any $\delta>0$, 
\[  \Pr\left[| \M(H_1, \dots, H_k) | \geq |\opt(G) | - O\left( k \sqrt{| \opt(G)| \ln \frac k \delta } \right)\right]\geq 1-\delta.
\]
\end{corollary}

As advertised, the mechanism underlying Corollary~\ref{cor:ir} is very simple: \emph{Choose an arbitrary optimal matching $\opt(G)$, independently of  $H_1, \dots, H_k$. If for all players $i\in[k]$ we have $|\opt(H_i)| \leq|\opt(G) \restriction H_i|$, then  $\M(H_1, \dots, H_k) = \opt(G)$. Else, let $\M(H_1, \dots, H_k) = \bigcup_{i} \opt(H_i)$.} We call this mechanism the \emph{Veto} mechanism, as any player can veto the proposed optimal matching. Alternatively, we can let players defect if they wish, while allowing the remaining players to continue to work together; for our mathematical purposes this is the same as the Veto mechanism, but the latter interpretation may be even more appealing from a practical viewpoint. 

The proof of Corollary~\ref{cor:ir} appears in Appendix~\ref{app:ir}. In a nutshell, the idea is that because $|\opt(H_i)|$ is concentrated around its expectation by Lemma~\ref{lem:concentration}, if some player wants to veto the proposed matching then it is likely that $\E[\opt(H_i)]$ is close to $|\opt(G)\restriction H_i|$, which is tightly concentrated around $|\opt(G)|/k$ by Hoeffding's inequality. But due to symmetry, this is true for all players, so the players can obtain on their own almost what they can obtain by collaborating.

We remark --- without proof --- that Corollary~\ref{cor:ir} still holds even if the probabilities, instead of being equal, are of the form $1/s^{t_i}$ for a fixed $s\in \mathbb{N}$, and possibly different $t_1,\ldots,t_n$. Similarly to Theorem~\ref{thm:main}, we conjecture that the statement actually holds for any $p_1,\ldots,p_n$ such that $p_i\geq 1/2$ for all $i\in [k]$. To prove this, one would need to strengthen Lemma~\ref{lem:f(p)<pOPT}, as discussed in \S\ref{sec:opt}. But now there is another difficulty: One would need to show that if $\E[\opt(H_i)]$ is close to $p_i|\opt(G)|$ for one player, then the same is true for all players --- in which case falling back to the internal matchings is almost optimal. In the symmetric case, this claim trivially holds, which is precisely why we assume that $p_i=1/k$ for all $i\in [k]$. 

While we require a relatively strong assumption on the probabilities, it is satisfying that the theorem's bound is asymptotically tight. To show this, we present and analyze an example of a graph with $n$ vertices where, with constant probability, \emph{any} individually rational matching is smaller than the optimal matching by $\Omega(\sqrt{n})$.

\begin{example}
\label{ex:3cycle}
Suppose that there are two players, each with probability $p= \frac 12$.
Consider the graph $G$ shown in Figure~\ref{fig:lower_bound_3-cycles}, which consists of four layers $A$, $B$, $C$ and $D$, each with $n/4$ vertices.
Any two layers of the  graph are fully connected if there is an edge between them according to Figure~\ref{fig:lower_bound_3-cycles}. That is, 
the edge set of this graph is such that \emph{any 3 vertices} from $A$, $B$, and $C$, respectively, form a directed $3$-cycle, and \emph{any $2$ vertices} from $C$ and $D$, respectively, form a directed $2$-cycle.
It is optimal to match the vertices in $A$, $B$, and $C$ via $3$-cycles, and therefore $|\opt(G) | = 3n/4$.

It is easy to show, using the standard deviation of the binomial distribution,
that the number of vertices a player owns in each layer deviates by $\pm\Theta(\sqrt{n})$ from its expectation (either larger or smaller) with constant probability.
Denote the number of vertices owned by player $1$ in layers A, B, C, and D by $a$, $b$, $c$, and $d$, respectively. We focus on the case where  $c\in\{n/8+\sqrt{n},\ldots,n/8+(3/2)\sqrt{n}\}$, $a$ and $b$ are both in $\{n/8-(3/2)\sqrt{n},\ldots,n/8-\sqrt{n}\}$, and $d\geq 3\sqrt{n}$  --- which happens with constant probability. Intuitively, player $1$ is doing well, because he owns significantly more than half of the vertices in layer $C$, which is especially important. 

In the foregoing case, $\opt(H_1)$ is obtained by taking $\min\{a,b\}$ 3-cycles between $A$, $B$, and $C$ (as many as possible), and then $c-\min\{a,b\}$ $2$-cycles between $D$ and the unmatched vertices of $C$. Therefore, 
$$| \opt(H_1) | = 3 \cdot \min\{a,b\} + 2 (c-\min\{a,b\}) = 2\cdot c+\min\{a,b\} = 3n/8 + \Theta(\sqrt{n}).
$$

On the other hand, consider some matching $\mathcal{M}$ with $x$ 3-cycles and $y$ 2-cycles, such that (without loss of generality) $x+y = n/4$ --- as the total number is constrained by the $n/4$ vertices in layer $C$. Note that under the optimal matching we have $x=n/4$, and $$|\opt(G)\restriction H_1| = a+b+c = 3n/8-\Theta(n).$$ More generally, we have that $|\mathcal{M}\restriction H_1| \leq a+b+c+y$. In order to guarantee that $\mathcal{M}$ is individually rational for player 1, we must close the gap between $|\opt(H_1)|$ and $|\mathcal{M}\restriction H_1|$, which implies that $y=\Omega(\sqrt{n})$. That is, we must sacrifice $\Omega(\sqrt{n})$ 3-cycles in favor of 2-cycles. But that means that $|\mathcal{M}|\leq |\opt(G)|-\Omega(\sqrt{n})$. 
\end{example}

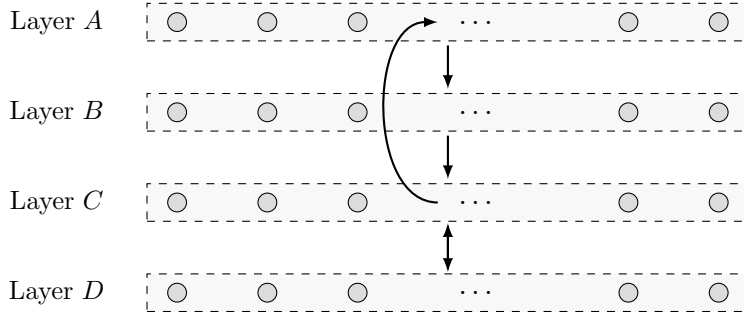
\begin{figure}
\centering
\begin{tikzpicture}[every shadow/.style={fill=black!30,shadow xshift=0.2ex,shadow yshift=-0.2ex},scale=0.8]
\tikzstyle{greycirc}=[circle,
draw=black,fill=mygray,thin,inner sep=0pt,minimum size=2.5mm]
\tikzset{Textbox/.style = {shape=rectangle}}
\tikzset{Trianglebox/.style = {draw,solid,thick,regular polygon, regular polygon sides=3,inner sep=2.5pt,fill=white, fill opacity=1.0}}
\tikzstyle{line}=[-]

\node[draw, rectangle, dashed, fill=black!3!,minimum width=80mm, minimum height=5mm] at (6.5,5.5) {};
\node (v11) at (2,5.5)  [greycirc] {};
\node (v12) at (3.5,5.5)  [greycirc] {};
\node (v13) at (5,5.5)  [greycirc] {};
\node (v14) at (9.5,5.5)  [greycirc] {};
\node (v15) at (11,5.5)  [greycirc] {};
\node (v16) at (7,5.5)  {\ldots};
\node (v17) at (6.5,5.5)  {};
\node[Textbox] at  (0,5.5) {\small Layer $A$};

\node[draw, rectangle, dashed, fill=black!3!,minimum width=80mm, minimum height=5mm] at (6.5,4) {};
\node (v21) at (2,4)  [greycirc] {};
\node (v22) at (3.5,4)  [greycirc] {};
\node (v23) at (5,4)  [greycirc] {};
\node (v24) at (9.5,4)  [greycirc] {};
\node (v25) at (11,4)  [greycirc] {};
\node (v26) at (7,4)  {\ldots};
\node (v27) at (6.5,4)  {};
\node[Textbox] at  (0,4) {\small Layer $B$};

\node[draw, rectangle, dashed, fill=black!3!,minimum width=80mm, minimum height=5mm] at (6.5,2.5) {};
\node (v31) at (2,2.5)  [greycirc] {};
\node (v32) at (3.5,2.5)  [greycirc] {};
\node (v33) at (5,2.5)  [greycirc] {};
\node (v34) at (9.5,2.5)  [greycirc] {};
\node (v35) at (11,2.5)  [greycirc] {};
\node (v36) at (7,2.5)  {\ldots};
\node (v37) at (6.5,2.5)  {};
\node[Textbox] at  (0,2.5) {\small Layer $C$};

\node[draw, rectangle, dashed, fill=black!3!,minimum width=80mm, minimum height=5mm] at (6.5,1) {};
\node (v41) at (2,1)  [greycirc] {};
\node (v42) at (3.5,1)  [greycirc] {};
\node (v43) at (5,1)  [greycirc] {};
\node (v44) at (9.5,1)  [greycirc] {};
\node (v45) at (11,1)  [greycirc] {};
\node (v46) at (7,1)  {\ldots};
\node (v47) at (6.5,1)  {};
\node[Textbox] at  (0,1) {\small Layer $D$};

\draw [line, black, -latex, shorten >=4pt, shorten <=5pt, line width=0.8pt] (v37) to (v47);
\draw [line, black, -latex, shorten >=4pt, shorten <=5pt, line width=0.8pt] (v47) to (v37);
\draw [line, black, -latex, shorten >=5pt, shorten <=5pt, line width=0.8pt] (v17) to (v27);
\draw [line, black, -latex, shorten >=5pt, shorten <=5pt, line width=0.8pt] (v27) to (v37);
\draw [line, black, -latex, in=180,out=180, line width=0.8pt] (v37) to (v17);

\end{tikzpicture} 
\caption{The graph constructed in Example~\ref{ex:3cycle}.}
\label{fig:lower_bound_3-cycles}
\end{figure}

Finally, note that Corollary~\ref{cor:ir} assumes cycles of constant length (as does Theorem~\ref{thm:main}). As noted in \S\ref{sec:intro}, major kidney exchanges do, in fact, only use very short cycles and chains (which can also be represented as cycles) in each match run. But it is nevertheless interesting to point out that the same statement is false when long cycles are allowed. Indeed, in Appendix~\ref{app:long} we present an example of a graph with long cycles, where (with high probability) every individually rational matching is smaller than the optimal matching by $\Omega(n)$.

\section*{Acknowledgments}
This work was partially supported by the NSF under grants DMS-1407558, IIS-1350598, CCF-1215883, CCF-1525932, CCF-1331175, and CCF-1525971; and by Caratheodory grant E.114 from the University of Patras, COST Action IC1205, Sloan Research Fellowship, IBM Ph.D.~Fellowship, and MSR Ph.D.~Fellowship.

\bibliographystyle{plain}
\bibliography{abb,ultimate}

\appendix

\section{Proof of Theorem~\ref{thm:main}: Omitted Claims}
\label{app:main}

This section contains proofs of claims that were omitted from the proof of Theorem~\ref{thm:main}. The claims themselves are stated in \S\ref{sec:opt}.

\subsection{Proof of Claim~\ref{claim:sum}}
\label{app:sum}

Let us define $\opt(G)\restriction^* G_i$ to be the vertices of $\opt(G)\restriction G_i$ that are matched by edges that lie within $G_i$. Since the decomposition is edge-disjoint, it holds that $|\opt(G)| = \sum_i|\opt(G)\restriction^* G_i|$.
It is therefore sufficient to show that for all $i$, 
$|\opt(G_i)| = |\opt(G)\restriction^* G_i|$. 
There are three cases:

\begin{enumerate}
\item $G_i$ corresponds to a component of $G[C]$. Recall that $\opt(G)$ matches all vertices of $C\subseteq B$. Moreover, $C$ has no edges to $D$, and $A$ is only matched with $D$, so the vertices of $G_i$ have no matched edges outside $G_i$. It follows that $\opt(G)\restriction^* G_i$ is itself a perfect matching on $G_i$, and
$|\opt(G_i)| = |\opt(G)\restriction^* G_i|$.
\item $G_i$ corresponds to a star with vertex $i\in A$: 
For each  $i\in A$, by the first property of the Edmonds-Gallai Decomposition, $i$ is matched to a distinct component of $G[D]$. Therefore, $\opt(G)\restriction^* G_i$ includes an edge from  $\opt(G)$. Since any star can have at most one matched edge, we have that $|\opt(G)\restriction^* G_i| = \opt(G)$.
\item $G_i$ corresponds to a component of $G[D]$: 
Since such a component is factor-critical, it has an odd number of vertices, and, for any vertex, a maximum matching that covers all other vertices. Therefore, both $\opt(G_i)$ and $\opt(G)\restriction^* G_i$ match all but one vertex of this component, and $| \opt(G_i)| =   |\opt(G)\restriction^* G_i|$.
\end{enumerate}
\qed

\subsection{Proof of Claim~\ref{claim:type23}}
\label{app:type23}

Let  $t\in\mathbb{N}$ and $p\in [0,1/2]$. It holds that
\begin{equation}
\label{eq:jeefa}
\frac{1}{2} -\frac{1}{2}(1-2p)^{2t+1}  \geq  p,
\end{equation}
because the left hand side is concave and has value $0$ for $p=0$ and $1/2$ for $p=1/2$. We also use the equalities
\begin{equation}
\label{eq:jeefa2}
\sum_{i=0}^k{{k \choose i} x^i} = (1+x)^k,
\end{equation}
and
\begin{equation}
\label{eq:jeefa3}
\sum_{i=1}^k{i{k \choose i}x^{i-1}} = k (1+x)^{k-1}.
\end{equation}

Assume that $n = 2t +1$ for some $t\geq 0$. By the claim's assumption, it holds that $\opt(G) = 2t$. 
Any matching among a set of $i$ vertices matches at most $2\lfloor i/2\rfloor$ vertices. Hence, the expected matching size of the subgraph induced by a random set of vertices when each vertex is included independently with probability $p$ is
\begin{eqnarray*}
\E_{H\sim_p G}\left[ |\opt(H)| \right] & \leq & \sum_{i=1}^{2t+1}{2\lfloor i/2\rfloor {2t+1 \choose i} p^i(1-p)^{2t+1-i}} \\
&=& p(1-p)^{2t}\sum_{i=1}^{2t+1}{i{2t+1 \choose i} \left(\frac{p}{1-p}\right)^{i-1}}-\frac{1}{2}(1-p)^{2t+1}\sum_{i=0}^{2t+1}{{2t+1 \choose i} \left(\frac{p}{1-p}\right)^i}\\
& & +\frac{1}{2}(1-p)^{2t+1}\sum_{i=0}^{2t+1}{{2t+1 \choose i} \left(\frac{p}{1-p}\right)^i(-1)^i}\\
&=& (2t+1)p-\frac{1}{2}+\frac{1}{2}(1-2p)^{2t+1}\\
&\leq & 2tp,
\end{eqnarray*}
where the penultimate transition follows by applying Equation~\eqref{eq:jeefa3} to the first term on the left hand side, and Equation~\eqref{eq:jeefa3} to the second and third terms; and the last transition follows from Equation~\eqref{eq:jeefa}.

If $n=2t$ and $\opt(G)\geq n-1$, it must hold that $\opt(G) = 2t$, because each edge corresponds to two matched vertices. Moreover, 
\begin{align*}
\E_{H\sim_p G}\left[ |\opt(H)| \right]  & \leq \sum_{i=1}^{2t}  {2\lfloor i/2\rfloor {2t \choose i} p^i(1-p)^{2t-i}} \\
       & \leq  p(1-p)^{2t - 1}\sum_{i=1}^{2t}{i{2t \choose i} \left(\frac{p}{1-p}\right)^{i-1}} \\
      &  =  p    (1-p)^{2t- 1} 2t\left(1 + \frac{p}{1-p} \right)^{2i - 1}\\
& = 2tp.
\end{align*}
\qed

\subsection{Proof of Lemma~\ref{lem:concentration}: Omitted Lower-Tail Bound}
\label{app:lower}

Our proof of Equation~\eqref{eq:lower} relies on the concept of \emph{self-bounding function}.
 
\begin{definition}\emph{\cite{boucheron2009concentration}}
A function $g:\X^n \rightarrow \R$ is $(a, b)$-self-bounding if
there exist functions $g_{-i}:\X^{n-1}\rightarrow \R$ for all $i\in [n]$ such that 
for all $\mathbf{x} = (x_1, \dots, x_n)\in \X^n$ and $i\in [n]$, 
\[     0\leq g(\mathbf{x} ) - g_{-i}(\mathbf{x}_{-i}) \leq 1,
\]
and 
\[  \sum_{i=1}^n \left( g(\mathbf{x} ) - g_{-i}(\mathbf{x} _{-i}) \right) \leq a \cdot g(\mathbf{x} ) + b,
\]
where $\mathbf{x}_{-i} = (x_1, \dots, x_{i-1}, x_{i+1}, \dots, x_n)$ is obtained by dropping the $i$th component of $\mathbf{x}$. 
\end{definition}
\begin{lemma}\label{prop:lowertail}\emph{\cite{boucheron2009concentration}}
If $Z = g(X_1, \dots, X_n)$, where $X_i\in \{0,1\}$ are independent random variables and $g$ is an $(a,b)$-self-bounding function with $a\geq \frac 13$, then for any $0<t<\E[Z]$,
\[ \Pr[Z\leq \E[Z] - t ] \leq \exp \left( - \frac{t^2}{2a\cdot \E[Z] + 2b}   \right).
\]
\end{lemma}
Let $g(\mathbf{x})$ be $\frac 1L$ times the size of optimal matching on the subgraph whose vertices correspond to the non-zero $x_i$'s, i.e., $g(\mathbf{x}) = \frac 1L| \opt(H_{\mathbf{x}})|$. Define $g_{-i}(\mathbf{x}_{-i}) = \min_{x_i}g(\mathbf{x})$.
We show that $g(\cdot)$ is $(L, 0)$-self-bounding.

Since $g_{-i}(\mathbf{x}_{-i})$ is the matching size in $H_{\mathbf{x}_{-i}}$ and 
$|\opt(H_{\mathbf{x}} ) |  \geq  | \opt(H_{\mathbf{x}_{-i}} ) | \geq  | \opt(H_{\mathbf{x}} )| - L$, we have that 
$0\leq g(\mathbf{x}) - g_{-i}(\mathbf{x}_{-i}) \leq 1.$
Furthermore,  $g(\mathbf{x}) - g_{-i}(\mathbf{x}_{-i})$ is non-zero only if vertex $i$ was in every $\opt(H_{\mathbf{x}})$. Since there are at most $|\opt(H_{\mathbf{x}})|$ such variables, we have
\[  \sum_{i=1}^n \left( g(\mathbf{x} ) - g_{-i}(\mathbf{x} _{-i}) \right) \leq |\opt(H_{\mathbf{x}})| 
=L\cdot  g(\mathbf{x}).
\]
Because $L > \frac 13$, we can use Lemma~\ref{prop:lowertail} with $\epsilon = t L$, and obtain
\begin{align*}
\Pr \left[  |\opt(H_i)| \leq  \E[|\opt(H_i)|]  - \epsilon  \right]  \leq \exp \left( - \frac{( \epsilon/L)^2}{2 L  (\frac 1L \E[|\opt(H_i)|])}   \right)
  \leq \exp \left( - \frac{\epsilon ^2}{2L^2 \E[|\opt(H_i)|]}   \right).
\end{align*}
\qed

\paragraph{Why self-bounding functions do not lead to better upper bounds.}
One might wonder 
whether the existing upper-tail bound of self-bounding functions could be used similarly to achieve an improved upper bound of  $L \sqrt{ 2\cdot \E[|\opt(H_i)|] \ln \frac{2k}{\delta} }$ for Lemma~\ref{lem:concentration} --- that is, a bound that depends on $|\opt(H_i)|$ instead of $|\opt(G_i)|$. Here, we answer this question in the negative.
The next lemma bounds the upper tail of $(a,b)$-self-bounding functions.
\begin{lemma}\label{prop:selfbounding}\emph{\cite{boucheron2009concentration}}
If $Z = g(X_1, \dots, X_n)$, where $X_i\in \{0,1\}$ are independent random variables and $g$ is an $(a,b)$-self-bounding function, then for any $0<t<\E[Z]$,
\[ \Pr[Z\geq \E[Z] + t ] \leq \exp \left( - \frac{t^2}{2a\cdot \E[Z] + 2b + 2ct}   \right),
\]
where $c = \max\{0, (3a-1)/6\}$.
\end{lemma}
Note that for $a = L > \frac 13$, the additional $2ct$  term in the denominator causes the
upper-tail bound to decay only as a simple exponential, and leads to significantly weaker concentration. Whether the upper-tail bound of self-bounding functions can be improved to remove this term is an open problem in probability theory, with the first bound appearing in the work of Boucheron et al.~\cite{boucheron2000sharp}, and improved bounds due to McDiarmid and Reed~\cite{mcdiarmid2006concentration} and Boucheron et al.~\cite{boucheron2009concentration}.
Successfully removing the $2ct$ term from the denominator would improve the result stated in Theorem~\ref{thm:main} from $O\left(\sqrt{|\opt(G)| \ln (k/\delta)}\right)$ to $O\left(\sqrt{p_i|\opt(G)| \ln (k/\delta)}\right)$, and Corollary~\ref{cor:ir} from  $O\left(k\sqrt{ |\opt(G)| \ln (k/\delta)}\right)$ to $O\left(\sqrt{k |\opt(G)| \ln (k/\delta)}\right)$.

\section{Proof of Corollary~\ref{cor:ir}}
\label{app:ir}

Let $p=1/k$. If 
$$
\E_{H_i\sim_{p} G}[\opt(H_i)] \leq p~| \opt(G)| -(2L+1)\sqrt{|\opt(G) \ln \frac{2k}{\delta}}
$$ 
then
{\small
\begin{align*}
&\Pr_{H_i \sim G} \left[  |\opt(H_i)| >|\opt(G) \restriction H_i| \right] \\
& \leq  \Pr_{H_i \sim G} \left[  |\opt(H_i)| - |\opt(G) \restriction H_i|  + \left(p~| \opt(G)| - \E[\opt(H_i)]- (2L+1)\sqrt{|\opt(G)| \ln \frac{2k}{\delta}}\right) > 0  \right] \\
&\leq  \Pr_{H_i \sim G} \left[  |\opt(H_i)|  - \E[\opt(H_i)] >2L\sqrt{|\opt(G)| \ln \frac{2k}{\delta}}\right] +\Pr_{H_i \sim G} \left[  p~| \opt(G)| -  |\opt(G) \restriction H_i| >  \sqrt{|\opt(G)| \ln\frac{2k}{\delta} }\right]\\
&\leq \frac \delta k,
\end{align*}}
where the last inequality holds by the upper-tail bound of Lemma~\ref{lem:concentration} and Hoeffding's inequality.
Therefore, with probability $1-\delta$, no player vetoes the proposed optimal matching, and $\M(H_1, \dots, H_k)  = \opt(G)$ is optimal.

On the other hand, if 
$$
\E_{H_i\sim_{p} G}[\opt(H_i)] \geq p~| \opt(G)| -  (2L+1)\sqrt{|\opt(G) \ln \frac{2k}{\delta}},
$$ 
then the expected total size of the internal matching is large, and we can fall back to the internal matchings. Indeed, note that $ \E_{H_i\sim_{p} G}[ | \opt(H_i) |]\leq p|\opt(G)|$ by symmetry. By the lower-tail bound of Lemma~\ref{lem:concentration}, with probability $1 - \delta$, for all $i\in [k]$,
\[ | \opt(H_i) |     \geq    \E_{H_i\sim_{p} G}[ | \opt(H_i) |] -  L \sqrt{2p |\opt(G) |\ln \frac{2k}{\delta} }.
\]
Therefore,  with probability $1-\delta$, 
\begin{align*}
\sum_{i=1}^k   |\opt(H_i)|  &\geq  k \E_{H_i\sim G}[ | \opt(H_i) |] -  k L \sqrt{ \frac 2k |\opt(G) |\ln \frac{2k}{\delta} }\\
&\geq | \opt(G)| - k(2L+1)\sqrt{|\opt(G) \ln \frac{2k}{\delta}} -  k L \sqrt{ \frac 2k |\opt(G) |\ln \frac{2k}{\delta} }\\
&= | \opt(G)| - O\left(k \sqrt{ |\opt(G) |\ln \frac{k}{\delta} } \right).
\end{align*}
So,  $\M(H_1, \dots, H_k)  = \bigcup_i{\opt(H_i)}$ is a near optimal. \qed
\section{Additional Examples}

In this section we present two examples that are referenced in the body of the paper.

\subsection{The Case of Large $p$}\label{app:large_p_example}
We provide examples where the conclusion of Theorem~\ref{thm:main} is violated under $p_i>1/2$. This happens because the examples violate Lemma~\ref{lem:f(p)<pOPT}, that is, they satisfy $$\E_{H\sim_p G}[ | \opt(H) |]  > \E_{H\sim_p G}[|\opt(G)\restriction H|]= p\cdot |\opt(G)|.$$ 

First, suppose that only 2-cycles are allowed. Consider a graph with three vertices $v_1,v_2,v_3$, and 2-cycles between $v_1$ and $v_2$, $v_2$ and $v_3$, and $v_3$ and $v_1$. Suppose $p=2/3$. Then $p |\opt(G)| = 2\cdot (2/3)$. On the other hand, $|\opt(H)|=2$ if $H$ contains at least two vertices (otherwise it is $0$), hence
$$
\E_{H\sim_p G}[ | \opt(H) |] = 2\left(\binom{3}{2} p^{2}(1-p) + \binom{3}{3} p^3\right) = 2\cdot \frac{20}{27} > 2\cdot\frac{2}{3}.
$$
Now consider a graph that contains many disjoint copies of the one just discussed. We have that both $\opt(H)$ and $|\opt(G)\restriction H|$ are concentrated around their expectations (by Hoeffding's inequality), so, with high probability, $|\opt(H)|>|\opt(G)\restriction H| + \Omega(n)$. 

When 3-cycles are allowed too, it is possible to show that the same phenomenon happens, for a value of $p$ sufficiently close to $1$, in a graph with five vertices $v_1,\ldots,v_5$, and 2-cycles between $v_i$ and $v_{i+1}$ for $i=1,\ldots,4$, as well as between $v_5$ and $v_1$.

\subsection{The Case of Long Cycles}
\label{app:long}

We construct an example where long cycles are allowed, and every individually rational matching is smaller than the optimal matching by $\Omega(n)$.  Motivated by kidney exchange, the example includes an altruistic donor, and matchings may include chains initiated by the altruist. However, we can easily transform the example into one where matchings can only include cycles, by adding directed edges from every vertex in the graph to the altruistic donor. 

\begin{figure}[t]
\centering
\begin{tikzpicture}[every shadow/.style={fill=black!30,shadow xshift=0.2ex,shadow yshift=-0.2ex}]
\tikzstyle{greycirc}=[circle,
draw=black,fill=mygray, thin,inner sep=0pt,minimum size=2.5mm]
\tikzset{Textbox/.style = {shape=rectangle}}
\tikzset{Trianglebox/.style = {draw,solid,thick,regular polygon, regular polygon sides=3,inner sep=0.5pt,fill=white, fill opacity=1.0}}
\tikzstyle{line}=[-]

\node[Trianglebox] (donor) at (0,1.75) {\small{$d$}};
\node[Textbox] at  (0,1.25) {\footnotesize Altruistic};
\node[Textbox] at  (0,0.95) {\footnotesize donor};

\node (v1) at (2,1)    [greycirc] {};
\node (v2) at (3.5,1)  [greycirc] {};
\node (v3) at (5,1)    [greycirc] {};
\coordinate[right=0.8cm of v3] (rv3);
\node (v4) at (8.5,1)  [greycirc] {};
\coordinate[left=0.8cm of v4]   (lv4);
\node (v5) at (10,1)   [greycirc] {};
\node (v6) at (11.5,1) [greycirc] {};
\draw[decoration={brace,mirror,raise=4pt,amplitude=6pt},decorate] (1.8,0.85) -- node[below=8pt] {\footnotesize Chain of length $\frac{3n}{9}$} (11.75,0.85);

\draw [line, black, -latex](v1) to (v2);
\draw [line, black, -latex](v2) to (v3);
\draw [line, black, -latex](v3) to (rv3);
\path (rv3) -- node[auto=false]{\ldots} (lv4);
\draw [line, black, -latex](lv4) to (v4);
\draw [line, black, -latex](v4) to (v5);
\draw [line, black, -latex](v5) to (v6);

\node[draw, ellipse, dashed, fill=black!3!,minimum width=5mm, minimum height=25mm] at (2,3.25) {};
\node (v11) at (2,2.5)  [greycirc] {};
\node (v12) at (2,3.25)    [greycirc] {};
\node (v13) at (2,4)  [greycirc] {};
\node[Textbox] at  (2,5.25) {\footnotesize Layer $1$};

\node[draw, ellipse, dashed, fill=black!3!,minimum width=5mm, minimum height=25mm] at (3.5,3.25) {};
\node (v21) at (3.5,2.5)  [greycirc] {};
\node (v22) at (3.5,3.25)    [greycirc] {};
\node (v23) at (3.5,4)  [greycirc] {};
\node[Textbox] at  (3.5,5.25) {\footnotesize Layer $2$};

\node[draw, ellipse, dashed, fill=black!3!,minimum width=5mm, minimum height=25mm] at (5,3.25) {};
\node (v31) at (5,2.5)  [greycirc] {};
\node (v32) at (5,3.25)    [greycirc] {};
\node (v33) at (5,4)  [greycirc] {};
\node[Textbox] at  (5,5.25) {\footnotesize Layer $3$};
\node (rv32) at (6.5,3.25)   {};

\node[draw, ellipse, dashed, fill=black!3!,minimum width=5mm, minimum height=25mm] at (8.5,3.25) {};
\node (v41) at (8.5,2.5)  [greycirc] {};
\node (v42) at (8.5,3.25)    [greycirc] {};
\node (v43) at (8.5,4)  [greycirc] {};
\node (lv42) at (7,3.25)  {};

\node[draw, ellipse, dashed, fill=black!3!,minimum width=5mm, minimum height=25mm] at (10,3.25) {};
\node (v51) at (10,2.5)  [greycirc] {};
\node (v52) at (10,3.25)    [greycirc] {};
\node (v53) at (10,4)  [greycirc] {};
\node[Textbox] at  (10,5.25) {\footnotesize Layer $\frac{2n}{9}$};

\draw [line, black, -latex, shorten >=5pt, shorten <=5pt](v12) to (v22);
\draw [line, black, -latex, shorten >=5pt, shorten <=5pt](v22) to (v32);
\draw [line, black, -latex, shorten >=5pt, shorten <=5pt](v32) to (rv32);
\draw [line, black, -latex, shorten >=5pt, shorten <=5pt](lv42) to (v42);
\draw [line, black, -latex, shorten >=5pt, shorten <=5pt](v42) to (v52);

\path (v32) -- node[auto=false]{\ldots} (v42);

\draw [line, black, -latex, shorten >=5pt, shorten <=5pt, out=20,in=160](v13) to (v33);
\draw [line, black, -latex, shorten >=5pt, shorten <=5pt, out=22,in=158](v13) to (v43);
\draw [line, black, -latex, shorten >=5pt, shorten <=5pt, out=25,in=155](v13) to (v53);

\draw [line, black, -latex, shorten >=5pt, shorten <=5pt, out=22,in=148](v23) to (v43);
\draw [line, black, -latex, shorten >=5pt, shorten <=5pt, out=25,in=145](v23) to (v53);

\draw [line, black, -latex, shorten >=5pt, shorten <=5pt, out=22,in=138](v33) to (v43);
\draw [line, black, -latex, shorten >=5pt, shorten <=5pt, out=25,in=135](v33) to (v53);

\draw [line, black, -latex, shorten >=5pt, shorten <=5pt, in=190,bend right = 7] (donor) to (v11);
\draw [line, black, -latex, shorten >=5pt, shorten <=5pt, in=190,bend right = 7] (donor) to (v21);
\draw [line, black, -latex, shorten >=5pt, shorten <=5pt, in=190,bend right = 7] (donor) to (v31);
\draw [line, black, -latex, shorten >=5pt, shorten <=5pt, in=190,bend right = 7] (donor) to (v41);
\draw [line, black, -latex, shorten >=5pt, shorten <=5pt, in=190,bend right = 7] (donor) to (v51);

\draw [line, black, -latex, shorten >=2pt, shorten <=2pt, bend right = 10] (donor) to (v1);

\end{tikzpicture}
\caption{The graph constructed in the example of Appendix~\ref{app:long}.}

\label{fig:long_chains_bad_example}
\end{figure}
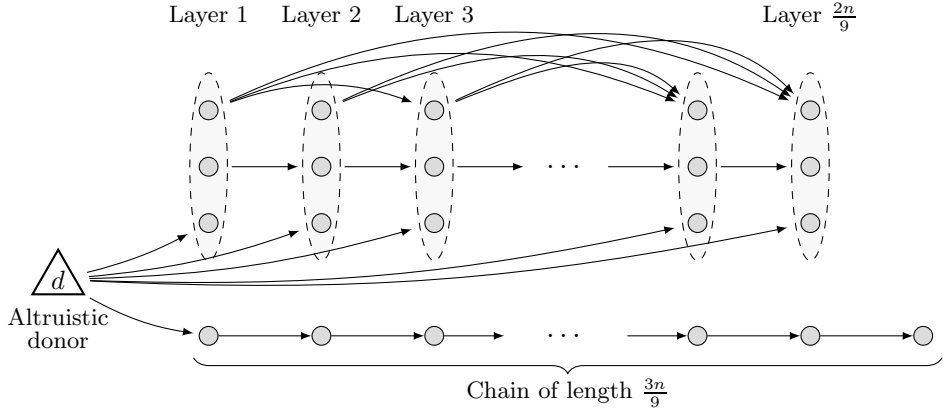

Consider the graph $G$ in Figure~\ref{fig:long_chains_bad_example}. The altruistic donor $d$ is shown as a triangle (we do not count him as one of the $n$ vertices). The vertices consist of (i) a chain of $\frac{3n}{9}$ vertices, and (ii) a network of $\frac{2n}{9}$ layers --- shown as dashed ellipses --- with each layer consisting of three vertices. All vertices in layer $i$ have edges to all vertices in layers $(i+1),\dots,\frac{2n}{9}$, and there are edges from $d$ to all vertices in each layer. Observe that since there are no cycles in this graph, all matches must happen only via a chain that originates in $d$. The longest chain consists of $\frac{3n}{9}$ vertices. Thus, there is a unique optimal matching in $G$, and its size is $|\opt(G)| = \frac{3n}{9}$. 

Let us now assume the there are two players with probability $p=\frac 12$ each. We first observe that the expected share of a subgraph $H \sim_p G$ in the optimal matching is $\E_{H \sim_p G}[ | \opt(G) \restriction H| ] = \frac{3n}{18}$. 

Next we examine $\opt(H)$. We can assume that $d\in H$, as this is always true for one of the two players (so we focus on that player without loss of generality). For any given assignment of the other vertices to the players, we say that a layer is \emph{good} if at least one of the vertices in that layer is in $H$. It is easy to see that --- under the assumption of $d\in H$ --- $\opt(H)$ is at least the number of good layers (via a chain that starts at $d$ and visits each good layer in order). Notice that each layer is good with probability $1 - (1/2)^3 = 7/8$. Therefore, the expected number of layers that are good is $\frac{7}{8} \cdot \frac{2n}{9}$. It follows that $\E_{H\sim p}[|\opt(H)|]\geq  \frac{14}{72}n$. 

Since both $|\opt(H)|$ and $|\opt(G) \restriction H|$ are almost always within $O(\sqrt{n})$ of their expected values, we have that with high probability, $|\opt(H)| - | \opt(G) \restriction H|  \geq \frac{2}{72}n -O(\sqrt{n})$. That is, $\opt(G)$ is not individually rational, and an individually rational matching would have to use $d$ to initiate a chain into the layered network. But such a chain can have length at most $\frac{2n}{9}$, whereas $\opt(G)=\frac{3n}{9}$ --- the difference is $\Omega(n)$, as desired.

\end{document}